\documentclass[journal]{IEEEtran}
\usepackage{pifont}
\usepackage{}
\usepackage{mathrsfs}
\usepackage[noadjust]{cite}
\usepackage{graphicx}
\usepackage{float}
\usepackage{subfigure}
\usepackage{amsthm}
\usepackage{amsmath}
\usepackage{bm}
\usepackage{enumerate}
\usepackage{amssymb}
\usepackage{fixltx2e}
\usepackage{color}
\usepackage{hhline}
\usepackage{algpseudocode}
\usepackage{amsmath}
\usepackage{graphics}
\usepackage{epsfig}
\usepackage{setspace}
\usepackage{multirow}
\usepackage{ragged2e}
\usepackage{makecell}
\usepackage{pifont}

\usepackage{bbding}
\usepackage[ruled]{algorithm2e}

\usepackage{graphics}
\usepackage{epsfig}
\newtheorem{lemma}{Lemma}

\allowdisplaybreaks[4]
\newtheorem{theorem}{Theorem}

\begin{document}

\title{AmBC-NOMA-Aided Short-Packet Communication for High Mobility V2X Transmissions}
\author{Xinyue~Pei, Xingwei~Wang, Yingyang~Chen,~\IEEEmembership{Senior Member,~IEEE,}\\Tingrui Pei, 	Miaowen~Wen,~\IEEEmembership{Senior Member,~IEEE}\\
	\thanks{X. Pei and X. Wang are with School of Computer Science and Engineering, Northeastern University, Shenyang 110819, China (e-mail:
		peixy@cse.neu.edu.cn, wangxw@mail.neu.edu.cn).}
	\thanks{Y. Chen and T. Pei are with the 
		College of Information Science and Technology, Jinan University, Guangzhou 510632, China (e-mail:
		chenyy@jnu.edu.cn; tingruipei@jnu.edu.cn). }
	\thanks{M. Wen is with the National Engineering Technology Research Center for Mobile Ultrasonic Detection, South China University of Technology, Guangzhou 510640, China (e-mail:
		eemwwen@scut.edu.cn).}
}
\markboth{SUBMITTED TO IEEE Wireless Communications Letters}{}
\maketitle

\begin{abstract}	
In this paper, we investigate the performance of  ambient backscatter communication  non-orthogonal multiple access (AmBC-NOMA)-assisted short packet communication for  high-mobility vehicle-to-everything  transmissions. In the proposed system, a roadside unit (RSU) transmits a superimposed signal to a typical NOMA user pair. Simultaneously,  the backscatter device (BD)  transmits its own signal towards the user pair by reflecting and modulating the RSU's superimposed signals. Due to vehicles' mobility, we consider 
 realistic assumptions of time-selective fading and channel estimation 
errors. Theoretical  expressions for the average block error rates (BLERs) of both users are derived. Furthermore, analysis and insights on  transmit signal-to-noise ratio, vehicles' mobility, imperfect channel estimation, the reflection efficiency at the BD, and blocklength  are provided. 
Numerical results validate the theoretical findings and reveal that the AmBC-NOMA system outperforms its orthogonal multiple access counterpart in terms of BLER performance. 
\end{abstract}

\maketitle
\section{Introduction}
With the development of intelligent transportation, vehicle-to-everything (V2X) communication has become extremely essential \cite{noor20226g}. 
However, as the number of vehicles has grown explosively, V2X needs to address several critical challenges,  such as massive access, spectrum and energy shortages.  Unfortunately, traditional  orthogonal multiple access (OMA)-based V2X networks may face congestion issues in such scenarios \cite{di2017v2x}. Compared with OMA,  the non-orthogonal multiple access (NOMA) technique, employing superposition coding at the transmitter and successive interference cancellation (SIC) at the receiver, has demonstrated its capability to enhance the spectral efficiency (SE) and meet the demands of massive connectivity   \cite{lai2019cooperative}. Consequently, integrating NOMA into V2X communications has gained significant attention.

To further meet the spectral and energy efficiency (EE) requirements of NOMA-based V2X systems, ambient backscatter communications (AmBC) has emerged as a  promising technology.
In AmBC systems, a backscatter device (BD) can reflect the signal transmitted by an ambient radio frequency (RF) source to the receiver, and it can overlap its own additional signal onto the existing RF signal through modulation, without the need for power-hungry active components \cite{liu2013ambient}. 
In this way, AmBC offers benefits such as high SE and EE, as well as flexible deployment \cite{yang2018cooperative}. Clearly, the combination of AmBC and NOMA is beneficial. In AmBC-NOMA, the BD's additional signal is decoded only by the near user through the SIC process, while it acts as interference for the far user \cite{zhang2019backscatter}.

Recent researches \cite{khan2021backscatter,zheng2022overlay,khan2022energy,peng2023ambient,ihsan2023energy} show that  AmBC-NOMA provides a feasible and promising solution in V2X. In \cite{khan2021backscatter}, the authors maximized the max-min achievable capacity of AmBC-NOMA V2X networks. In \cite{zheng2022overlay}, the authors elaborated on the secrecy performance for cognitive AmBC-NOMA V2X networks. In \cite{khan2022energy}, the authors investigated the EE optimization problem for AmBC-NOMA V2X communication  with imperfect channel state information (CSI) estimation. Furthermore, in \cite{peng2023ambient}, the authors studied the covertness performance of an AmBC-NOMA vehicular network. Finally, in \cite{ihsan2023energy}, the authors optimized the EE for AmBC-NOMA V2X sensor communications with imperfect CSI estimation.

Notably, prior works did not consider high-mobility scenarios. Unfortunately, in high-mobility scenarios, the mobility of vehicles and Doppler spread effects result in both imperfect and outdated CSI estimations, as well as time-selective fading \cite{khattabi2015performance}. Furthermore, existing works assumed a conventional infinite blocklength (IBL) transmission regime. Nevertheless, IBL codes are no longer appropriate for high-mobility scenarios \cite{xia2023noma} considering the inherent requirements for ultra-reliablity and low-latency. Instead, short-packet communication (SPC) with finite blocklength (FBL) codes has emerged as a physical-layer solution. Unlike long-packet  and IBL communication, for SPC, Shannon's channel capacity becomes inaccurate, and the block error rate (BLER) cannot be ignored \cite{durisi2016toward}. Hence, it is pivotal to analyze the BLER performance in AmBC-NOMA high-mobility V2X systems.

Apparently, the study of the FBL performance in the context of AmBC-NOMA high-mobility V2X systems is still in its infancy.  To address this gap,  we  evaluate the BLER performance in AmBC-NOMA systems over time-selective fading channels.  The key contributions are listed as follows:
\begin{itemize}
	\item  We explore the application of AmBC-NOMA-assisted SPC in high-mobility V2X systems, considering  imperfect and outdated estimation processes as well as time-selective fading. 
	\item We study the statistics of the received signal-to-interference-plus-noise ratios (SINRs), extracting closed-form expressions for the cumulative distribution function (CDF) of received SINRs.
	Then, utilizing those expressions, we first derive theoretical expressions for the average BLERs in the proposed systems.
		\item The accuracy of the derived results is corroborated through simulations. Additionally, numerical results investigate the impact of various key parameters.  Furthermore, the simulations indicate that our proposed AmBC-NOMA system achieves superior BLER performance compared to AmBC-OMA.
\end{itemize}
\emph{Notation}: The operations $\Pr(\cdot)$, $|\cdot|$,  and $\mathbb{E}\{\cdot\}$ denote the probability, the absolute value,  and the expectation,
	respectively. $F_X(\cdot)$ and $f_X(\cdot)$ respectively denote the CDF and probability density function (PDF) of a random variable $X$.
	A complex Gaussian distribution $Y$ with zero mean and variance $\Omega$ is represented by $Y\sim\mathcal{CN}(0,\Omega)$, with $F_{|Y|^2}(x)=1-\exp(-\frac{x}{\Omega})$ and $f_{|Y|^2}(x)=\frac{1}{\Omega}\exp(-\frac{x}{\Omega})$.   $Q(x)=\int_{x=0}^{\infty}\frac{1}{2\pi}\exp(\frac{-t^2}{2})$ is  the Gaussian Q-function, $K_v(\cdot)$ is the modified bessel function of the second kind,
	 and $\text{Ei}(\cdot)$ is the exponential integral function  \cite{gradshteyn2014table}.

\section{System Model}

In this work, we consider an AmBC-NOMA-assisted  high-mobility V2X scenario as depicted in Fig. \ref{system_model}, which consists of a roadside unit (RSU), a BD, a near vehicular user (denoted by $U_N$), and a far vehicular user (denoted by $U_F$).\footnote{Due to the strong interference between users, it may be challenging to jointly apply NOMA to all users in practice. A feasible alternative is to divide users into orthogonal pairs, where  
	each pair performs NOMA. Therefore, this letter focuses on  a typical 
	NOMA user pair to illustrate the effectiveness of the
	proposed scheme.} All the nodes are equipped with single antenna.

\subsection{High-mobility and Time-selective Channel Modeling}
It is assumed that all channels are subject to Rayleigh fading,\footnote{This assumption is widely used in AmBC-NOMA V2X networks \cite{khan2021backscatter,zheng2022overlay,khan2022energy,peng2023ambient,ihsan2023energy}.} and the channel coefficients from RSU to $U_N$, $U_F$, and BD are respectively denoted by $h_{RN}\sim\mathcal{CN}(0,\Omega_{RN})$, $h_{RF}\sim\mathcal{CN}(0,\Omega_{RF})$, and $h_{RB}\sim\mathcal{CN}(0,\Omega_{RB})$. Similarly, the channel coefficients from BD to $U_N$ and $U_F$ are respectively denoted by $h_{BN}\sim\mathcal{CN}(0,\Omega_{BN})$ and $h_{BF}\sim\mathcal{CN}(0,\Omega_{BF})$. The aforementioned $\Omega$ represents the average channel power.

For convenience, we assume that users are driving in the same direction with speed $v$ km/h on a
 highway.
Due to the high-mobility nature of users, the corresponding fading channel  $h_k$ ($k\in\{RN,RF,BN,BF\}$) is assumed to be time-selective.\footnote{
	Notably, since the BD is pre-positioned at a fixed location, we assume that $h_{RB}$ is static, so that we do not consider imperfect CSI for it. } To model the channel variation over time, the first-order auto-regressive process \cite{khattabi2015performance} is employed, and $h_k(t)$ over the $t$-th time instant can be described as
\begin{small}
	\begin{align}\label{AR1model}
		h_k(t)=\rho_kh_k(t-1)+\sqrt{1-\rho_k^2}e_k(t),
	\end{align}
\end{small}where $e_k(t)\sim\mathcal{CN}(0,\Omega_{e k})$ is the time-varying component of the channel, and $\rho_k$ is the correlation parameter according to Jakes' model. Specifically,  we have
$
	\rho_k=\mathcal{J}_0\left({2\pi f_D T_s}\right)
$.
Here $\mathcal{J}_0(\cdot)$ is the zeroth-order Bessel function of the first kind, $T_s$ is the transmitted symbol duration, and $f_D=f_cv/c$ is the Doppler frequency shift for the vehicle with speed $v$, where $f_c$ represents the carrier frequency and $c$ is the speed of light.

Time-selective $h_k$ undergoes rapid changes, posing significant challenges for receivers to achieve perfect channel estimation.  Furthermore,  real-time tracking of these channels is difficult, making it hard for receivers to estimate the channel at every time instant.
 Therefore, in this scenario, we assume that $U_k$ only estimate the CSI  at the first time instant of each coherence time \cite{khattabi2015performance}. Using the minimum mean square error (MMSE) model, $h_k(1)$ can be expressed as  
\begin{small}
\begin{align}\label{estimate_equation}
	h_k(1)=\hat{h}_k(1)+\epsilon_k(1),
\end{align}
\end{small}where
$\hat{h}_k(1)\sim\mathcal{CN}(0,\hat{\Omega}_{k})$ is the estimated CSI, $\epsilon_k(1)\sim\mathcal{CN}(0,\Omega_{\epsilon k})$ is the estimation error. With the aid of (\ref{AR1model}) and (\ref{estimate_equation}), we can achieve \cite{khattabi2015performance,xia2023noma}:
\begin{small}
\begin{align}\label{hk_define}
	h_k(t)=\underbrace{\rho_k^{t-1}\hat{h}_k(1)+\varphi_k(t)}_{\hat{h}_k(t)}+\underbrace{\rho_k^{t-1}\epsilon_k(1)}_{\epsilon_k(t)},
\end{align}
\end{small}where  $
	\varphi_k(t)=\sqrt{1-\rho_k^2}\sum_{i=1}^{t-1}\rho_k^{t-1-i}e_k(m) $ represents user mobility noise and  follows the distribution $\mathcal{CN}(0,(1-\rho_k^{2(t-1)})\Omega_{e k})$.

For calculation convenience, we define $\xi_k\triangleq \varphi_k(t)+\rho_k^{t-1}\epsilon_k(1)$ as the effective noise caused by estimation error and user mobility. Obviously, $\xi_k$ is a complex Gaussian variable following $\mathcal{CN}(0,\Omega_{\xi k})$,  where $\Omega_{\xi k}\triangleq(1-\rho_k^{2(t-1)})\Omega_{e k}+\rho_k^{2(t-1)}\Omega_{\epsilon k}$, then (\ref{hk_define}) be rewritten as
\begin{small}
\begin{align}\label{hk_simple}
	h_k(t)=\rho_k^{t-1}\hat{h}_k(1)+\xi_k.
\end{align}
\end{small}

\begin{figure}[t]
	\centering
	\includegraphics[width=2.2in]{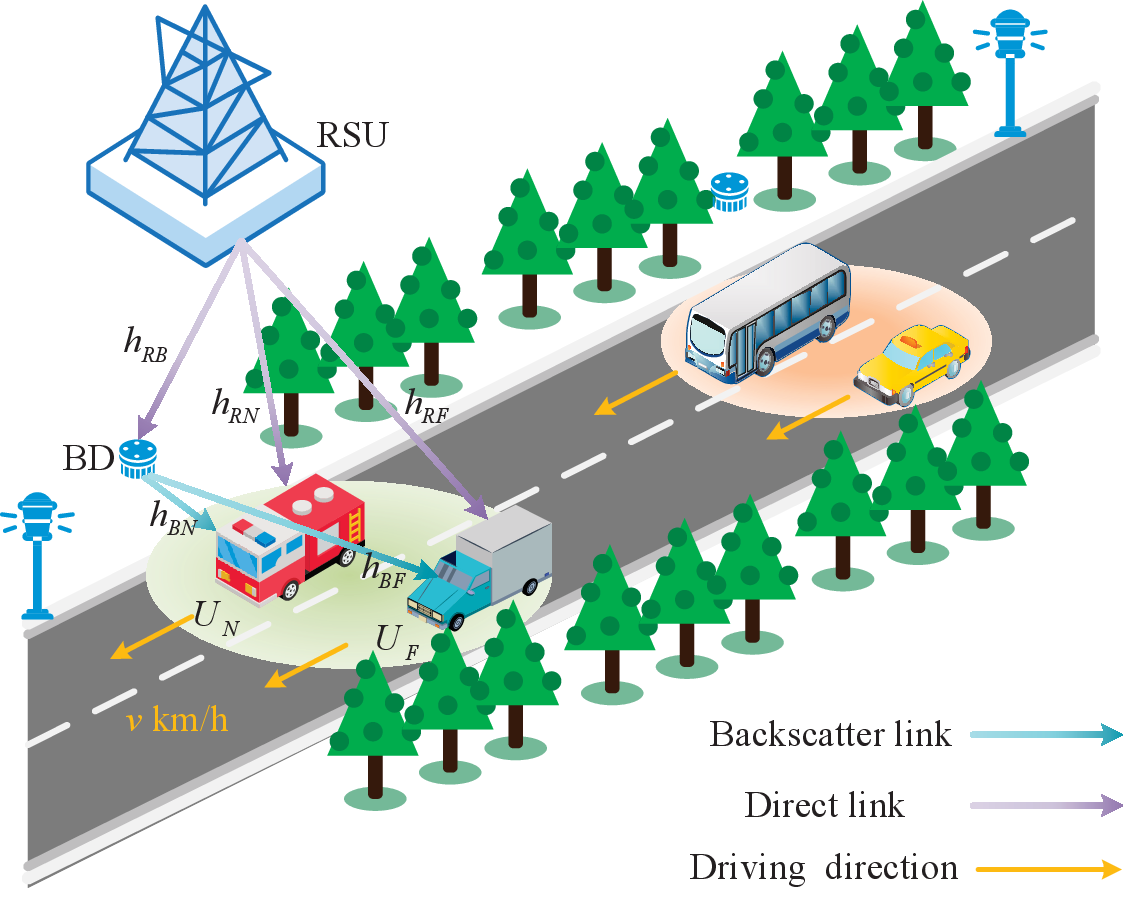}
	\caption{System model.} 	
	\label{system_model}
\end{figure}
\vspace*{-0.5cm}
\subsection{Transmission Model}
Utilizing the NOMA protocol, the RSU transmits a superimposed signal $s(t)=\sqrt{Pa_N}s_N(t)+\sqrt{Pa_F}s_F(t)$ to $U_N$ and $U_F$, where $P$ is the total transmit power; $s_N$ and $s_F$ respectively denote the intended signals for $U_N$ and $U_F$, with $\mathbb{E}\{|s_N|^2\}=\mathbb{E}\{|s_F|^2\}=1$; $a_N$ and $a_F$ are respectively the power coefficients allocated to $s_N$ and $s_F$ with $a_N+a_F=1$ and $a_F>a_N$. 

At the same time, the BD backscatters the RSU’s signal $s$ to users with its own message $s_C$, where $\mathbb{E}\{|s_C|^2\}=1$. 
The signals received by $U_N$ and $U_F$ are respectively
given by
\begin{small}
\begin{align}
	y_N=h_{RN}(t)s(t)+h_{RB}h_{BN}(t)\beta s(t)s_C(t)+n_N,
\end{align}
\end{small}and
\begin{small}
\begin{align}
	y_F=h_{RF}(t)s(t)+h_{RB}h_{BF}(t)\beta s(t)s_C(t)+n_F,
\end{align}
\end{small}where $|\beta|\leq 1$ denotes the reflection efficiency of BD; $n_N$ and $n_F\sim\mathcal{CN}(0,\sigma^2)$ respectively represent the additive white
Gaussian noises (AWGNs) at $U_N$ and $U_F$.

Unlike  conventional NOMA, under the AmBC-NOMA protocol,  $U_N$ first decodes $s_F$, then decodes $s_N$, and finally decodes $s_C$. The above process can be achieved by using SIC technique. The received SINRs of $s_F$, $s_N$, and $s_C$ at $U_N$ can be respectively written as
\begin{small}
\begin{align}
	&\gamma_{U_N}^{s_F}\nonumber\\
		&=\frac{\mathcal{A}_N|\hat{h}_{RN}(1)|^2\gamma}{\gamma\left[
			\mathcal{G}_N|\hat{h}_{RN}(1)|^2 \!\!+\!\!\Omega_{\xi RN}\!\!+\!\!|h_{RB}|^2(\mathcal{C}_N|\hat{h}_{BN}(1)|^2\!\!+\!\!\mathcal{D}_N)\right]\!\!+\!\!1},
\end{align}
\end{small}
\begin{small}
	\begin{align}
\gamma_{U_N}^{s_N}=\frac{\mathcal{G}_N|\hat{h}_{RN}(1)|^2 \gamma}{\gamma\left[ \Omega_{\xi RN}+|h_{RB}|^2(\mathcal{C}_N|\hat{h}_{BN}(1)|^2+\mathcal{D}_N)\right]+1},
	\end{align}
\end{small}and
\begin{small}
	\begin{align}\label{define_gammaUNSC}
		\gamma_{U_N}^{s_C}=\frac{|h_{RB}|^2|\hat{h}_{BN}(1)|^2\mathcal{C}_N\gamma}{\gamma\left( \Omega_{\xi RN}+|h_{RB}|^2\mathcal{D}_N\right)+1},
	\end{align}
\end{small}where $\gamma=P/\sigma^2$ represents transmit signal-to-noise ratio (SNR), $\mathcal{A}_{N}=\rho_{RN}^{2(t-1)} a_F$, $\mathcal{G}_N=\rho_{RN}^{2(t-1)}a_N$, $\mathcal{C}_N=\rho_{BN}^{2(t-1)}\beta^2$, and $	\mathcal{D}_N=\Omega_{\xi RN}\beta^2$.

Meanwhile, $U_F$ only decodes $s_F$ by treating other signals as noise, whose SINR can be written as
\begin{small}
	\begin{align}
		&\gamma_{U_F}^{s_F}\nonumber\\
		&=\frac{\mathcal{A}_F|\hat{h}_{RF}(1)|^2 \gamma}{\gamma\left[
			\mathcal{G}_F|\hat{h}_{RF}(1)|^2 +\Omega_{\xi RF}+|h_{RB}|^2(\mathcal{C}_F|\hat{h}_{BF}(1)|^2+\mathcal{D}_F)\right]+1},
	\end{align}
\end{small}where $\mathcal{A}_{F}=\rho_{RF}^{2(t-1)} a_F$,  $\mathcal{G}_F=\rho_{RF}^{2(t-1)}a_N$, $\mathcal{C}_F=\rho_{BF}^{2(t-1)}\beta^2$, and $	\mathcal{D}_F=\Omega_{\xi BF}\beta^2$.

\begin{figure}[t]
	\centering
	\includegraphics[width=2.3in]{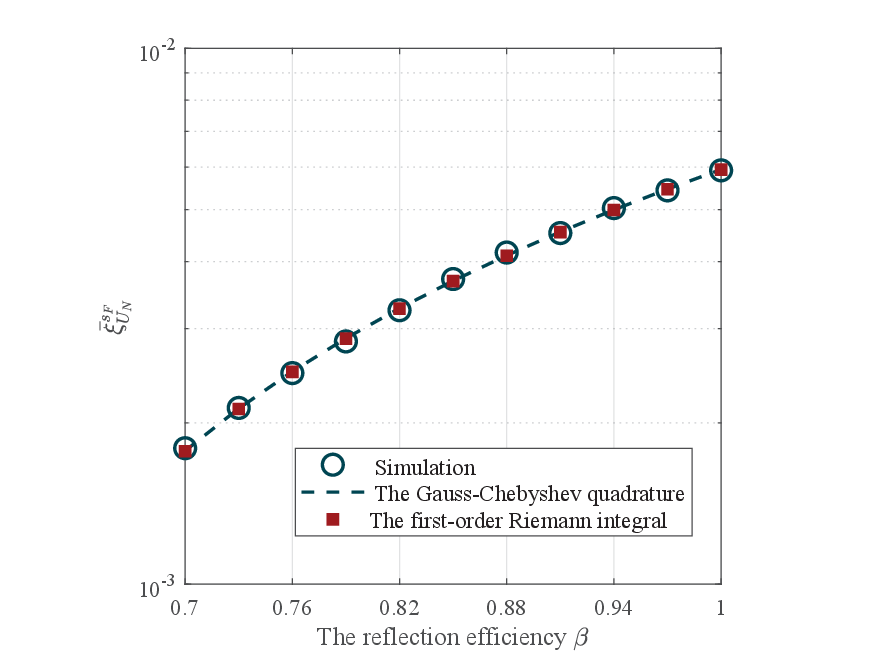}
	\caption{Illustration of the accuracy of the Gauss-Chebyshev quadrature and the first-order Riemann integral  approximations.} 	
	\label{accuracy_approximation}
\end{figure}

\begin{figure*}[htbp] 
	\begin{minipage}[t]{0.34\linewidth} 
		\centering
		\includegraphics[width=2.3in]{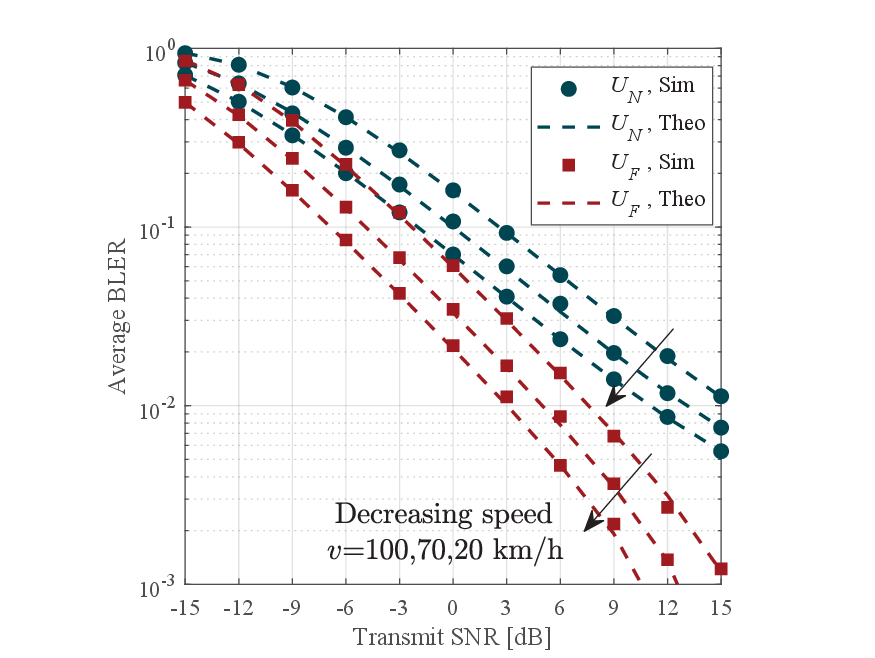} 
		\caption{Average BLERs  versus transmit SNR.} 	
		\label{abler_vs_SNR}
	\end{minipage}%
	\begin{minipage}[t]{0.34\linewidth}
		\centering
		\includegraphics[width=2.3in]{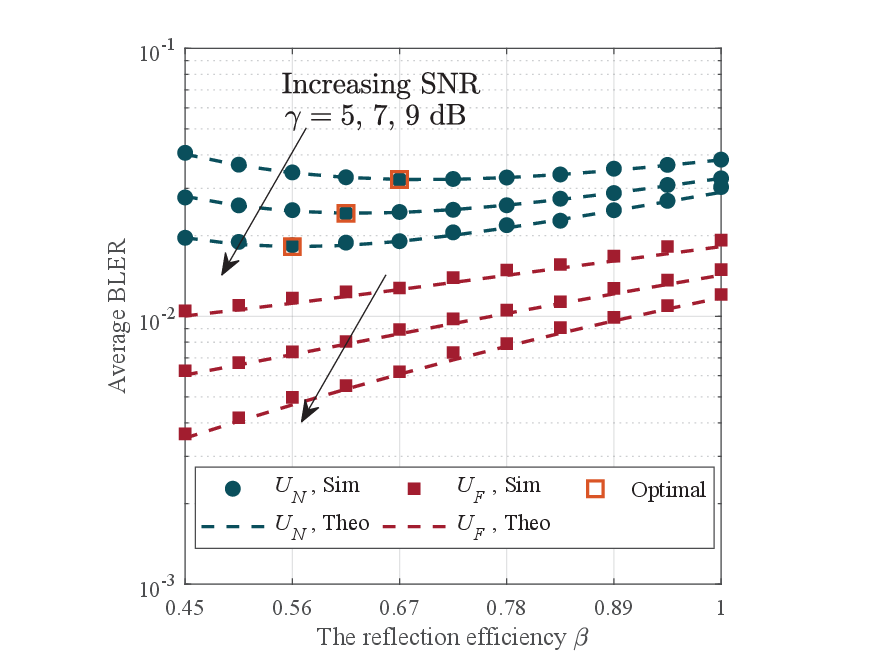}
		\caption{Average BLERs  versus reflection efficiency.} 	
		\label{abler_vs_beta}
	\end{minipage}%
	\begin{minipage}[t]{0.34\linewidth}
		\centering
		\includegraphics[width=2.3in]{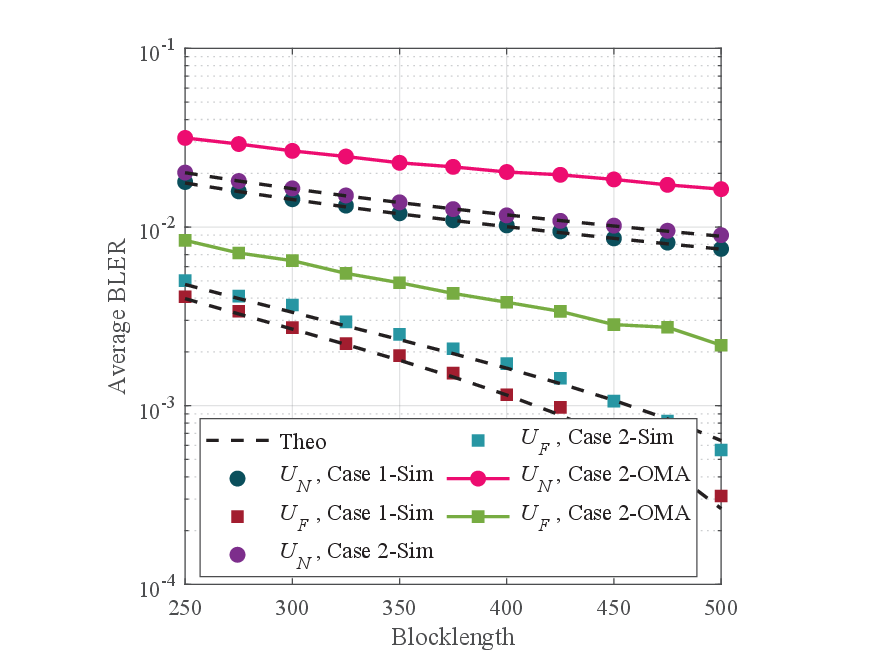}
		\caption{Average BLERs versus blocklength.} 	
		\label{abler_vs_blocklength}
	\end{minipage}
\end{figure*}

\vspace*{-0.3cm}
\section{The Statistics of SINRs}
To obtain the theoretical expressions for BLER, we have to first explore the statistics of  received SINRs. Specifically,  the exact expressions of CDFs for different SINRs are summarized in the following lemmas.
\begin{lemma}\label{CDF_of_gammaUm_sF}
	The CDF of received SINR  $\gamma_{U_m}^{s_F}$ ($m\in\{N,F\}$) can be derived as
	\begin{small}
		\begin{align}\label{F_gammaUm_sF}
	F_{\gamma_{U_m}^{s_F}}(u)=\begin{cases}
			1+\chi_{1m}(u){\rm Ei}(-\chi_{2m}(u))\exp(\chi_{3m}(u)),& u<\frac{\mathcal{A}_m}{\mathcal{G}_m},\\
				1,& u\geq\frac{\mathcal{A}_m}{\mathcal{G}_m},
			\end{cases}			
		\end{align}
	\end{small}where	 $	\chi_{1m}(u)\!=\!\frac{\hat{\Omega}_{Rm}(\mathcal{A}_m-u\mathcal{G}_m)}{u\Omega_{RB}\hat{\Omega}_{Bm}\mathcal{C}_m}$, $	\chi_{2m}(u)\!=\!\chi_{1m}(\frac{u\Omega_{RB}\mathcal{D}_m}{\hat{\Omega}_{Rm}(\mathcal{A}_m-u\mathcal{G}_m)}+1)$, and
	$\chi_{3m}(u)\!\!=\!\!\chi_{2m}-\frac{u\Omega_{\xi Rm}}{\hat{\Omega}_{Rm}(\mathcal{A}_m-u\mathcal{G}_m)}-\frac{u}{\hat{\Omega}_{Rm}\gamma(\mathcal{A}_m-u\mathcal{G}_m)}.\label{chi3}$ 
\end{lemma}
\begin{proof}
	See Appendix \ref{Proof_of_CDF_of_gammaUm_sF}.
\end{proof}
\begin{lemma}\label{CDF_of_gammaUN_sN}
	The CDF of received SINR $\gamma_{U_N}^{s_N}$ can be derived as
	\begin{small}
		\begin{align}\label{F_gammaUN_sN}
			F_{\gamma_{U_N}^{s_N}}(u)=	1+\Theta_{1N}(u){\rm Ei}(-\Theta_{2N}(u))\exp(\Theta_{3N}(u)),		
		\end{align}
	\end{small}where $	\Theta_{1N}(u)=\frac{\hat{\Omega}_{RN}\mathcal{G}_N}{u\Omega_{RB}\hat{\Omega}_{BN}\mathcal{C}_N}$, $	\Theta_{2N}(u)=\Theta_{1N}(\frac{u\Omega_{RB}\mathcal{D}_N}{\hat{\Omega}_{RN}\mathcal{G}_N}+1)$, and $	\Theta_{3N}(u)=\Theta_{2N}-\frac{u\Omega_{\xi RN}}{\hat{\Omega}_{RN}\mathcal{G}_N}-\frac{u}{\hat{\Omega}_{RN}\gamma\mathcal{G}_N}.\label{Theta3}$
\end{lemma}
\begin{proof}
The proof  is similar to  Appendix \ref{Proof_of_CDF_of_gammaUm_sF}, and thus omitted here due to space limitation.
\end{proof}
\begin{lemma}\label{CDF_of_gammaUN_sC}
	The CDF of received SINR $\gamma_{U_N}^{s_C}$ can be derived as
	\begin{small}
		\begin{align}\label{F_gammaUN_sC}
			F_{\gamma_{U_N}^{s_C}}(u)=	1-\varpi_{1N}(u)K_1(\varpi_{1N}(u))\exp(-\varpi_{2N}(u)),		
		\end{align}
	\end{small}where $\varpi_{1N}(u)=2\sqrt{\frac{u(\gamma\Omega_{\xi RN}+1)}{\hat{\Omega}_{BN}\Omega_{RB}\mathcal{C}_N\gamma}}$ and $\varpi_{2N}(u)=\frac{u\mathcal{D}_N}{\hat{\Omega}_{BN}\mathcal{C}_N}$.
\end{lemma}
\begin{proof}
	See  Appendix \ref{Proof_of_CDF_of_gammaUN_sC}.
\end{proof}

\section{BLER Analysis}
In this section, we present the BLER analysis of the proposed system. 
 For convenience, $R_{s_i}$ bit/s/Hz  ($i\in\{C,N,F\}$)  is used to denote the maximum achievable rate of $s_i$, with a blocklength of $\mathcal{L}_{s_i}$ bits. The instantaneous BLER corresponding to the SINR $\gamma_{U_m}^{s_i}$ ($m\in\{N,F\}$) decoding $s_i$  in SPC ($\mathcal{L}_{s_i}>100$ bits) can be expressed as \cite{yang2014quasi}
 \begin{small}
\begin{align}\label{instantaneous_BLER}
	\xi_{U_m}^{s_i}=Q\left(\frac{\mathcal{C}(\gamma_{U_m}^{s_i})-R_{s_i}}{\sqrt{V(\gamma_{U_m}^{s_i})/\mathcal{L}_{s_i}}}\right),
\end{align}
 \end{small}where $\mathcal{C}(\gamma_{U_m}^{s_i})=\log_2(1+\gamma_{U_m}^{s_i})$ is the Shannon capacity, and $V(x)=1/{(\ln 2)^2}\left(1-{1}/{(1+x)^2}\right)$ represents the channel dispersion.  Since obtaining a closed-form solution for the current form of $\xi_{U_m}^{s_i}$ is highly challenging, we adopt the linear approach described in \cite{makki2014finite} to tightly approximate the above $Q$-function, i.e., (\ref{instantaneous_BLER}), as
 \begin{small}
\begin{align}
	Z_{k,i}(\gamma_{U_m}^{s_i})=\begin{cases}
		1,&\gamma_{U_m}^{s_i}<\vartheta_i,\\
		0.5-\eta_i\sqrt{\mathcal{L}_{s_i}}(\gamma_{U_m}^{s_i}-\psi_i),&\gamma_{k,i}\in[\vartheta_i,\varsigma_i],\\
		0,&\gamma_{U_m}^{s_i}>\varsigma_i,
	\end{cases}
\end{align}
 \end{small}where
$\eta_i=1/\sqrt{2\pi(2^{2R_{s_i}}-1)}$, $\psi_i=2^{R_{s_i}}-1$,
$ \vartheta_i =\psi_i-1/(2\eta_i\sqrt{\mathcal{L}_{s_i}})$, and
$\varsigma_i=\psi_i+1/(2\eta_i\sqrt{\mathcal{L}_{s_i}})$. Then, the average BLER can be calculated as
\begin{small}
\begin{align}\label{BLER_integral}
	&\bar{\xi}_{U_m}^{s_i}=\mathbb{E}[\xi_{U_m}^{s_i}]=\int_{0}^{\infty}\!\!\xi_{U_m}^{s_i}f_{\gamma_{U_m}^{s_i}}(x)\text{d}x\nonumber\\
	&\approx\int_{0}^{\infty}\!\!Z_{\gamma_{U_m}^{s_i}}(x) f_{\gamma_{U_m}^{s_i}}(x)\text{d}x\!=\!\eta_{i}\sqrt{\mathcal{L}_{s_i}}\int_{\vartheta_i}^{\varsigma_i}F_{\gamma_{U_m}^{s_i}}(x)\text{d}x.
\end{align}
\end{small}Recall that $F_{\gamma_{U_m}^{s_i}}(x)$ involves the product of ${\rm Ei}(\cdot)$ (or $K_1(\cdot)$) and $\exp(\cdot)$, making it challenging to obtain a closed-form expression for the integral. To evaluate this integral, we can employ the first-order Riemann integral\footnote{This method can be used when the integral interval $\varsigma_i-\vartheta_i=\sqrt{\frac{2\pi}{\mathcal{L}_{s_i}}(2^{2R_{s_i}}-1)}$ is small. } \cite{ho2021short}, i.e.,   $\int_{a}^{b}f(x)\text{d}x\approx(b-a)f\left({(a+b)}/{2}\right)$, then (\ref{BLER_integral}) becomes
$
	\bar{\xi}_{U_m}^{s_i}\approx F_{\gamma_{U_m}^{s_i}}\left(\psi_i\right),
$
or  we can employ  the Gauss-Chebyshev quadrature method \cite{gradshteyn2014table}, resulting in
\begin{small}
\begin{align}
	\bar{\xi}_{U_m}^{s_i}\approx\sum_{i=1}^{n}\frac{\pi}{2n}F_{\gamma_{U_m}^{s_i}}\left(\frac{z_i+1}{2\eta_i\sqrt{\mathcal{L}_{s_i}}}+\vartheta_i\right)\sqrt{1-z_i^2},
\end{align}
\end{small}where $n$ represents the calculation accuracy, and
$
	z_i=\cos\left(\frac{2i-1}{2n}\pi\right).
$

The accuracy of these two approximations  is demonstrated in Fig. \ref{accuracy_approximation}, shown at the top of  previous page. It can be seen that both approximations exhibit high precision. Therefore, we herein adopt the simpler approximation method, i.e., the  first-order Riemann integral, to obtain the approximations.  Subsequently,  we can achieve the following theorems.
\begin{theorem}\label{ABLER_UN}
Based on the SIC operation,	the approximated expression for end-to-end (e2e) average   BLER at $U_N$ can be derived as
\begin{small}
	\begin{align}\label{varepsilonU_N_case1}
		\varepsilon_{U_N}\approx&	1-\chi_{1N}(\psi_F)\Theta_{1N}(\psi_N)\varpi_{1N}(\psi_C){\rm Ei}(-\chi_{2N}(\psi_F))\nonumber\\
		&\times{\rm Ei}(-\Theta_{2N}(\psi_N))K_1(\varpi_{1N}(\psi_C))\nonumber\\
		&\times\exp\left(\chi_{3N}(\psi_F)+\Theta_{3N}(\psi_N)-\varpi_{2N}(\psi_C)\right),	
	\end{align}
\end{small}when $ \psi_F<\frac{\mathcal{A}_m}{\mathcal{G}_m}$; otherwise, $\varepsilon_{U_N}$ can be derived as
\begin{small}
	\begin{align}\label{varepsilonU_N_case2}
		\varepsilon_{U_N}\approx&	1-\Theta_{1N}(\psi_N)\varpi_{1N}(\psi_C){\rm Ei}(-\Theta_{2N}(\psi_N))\nonumber\\
		&\times K_1(\varpi_{1N}(\psi_C))\exp\left(\Theta_{3N}(\psi_N)-\varpi_{2N}(\psi_C)\right),
	\end{align}
\end{small}when $ \psi_F\geq\frac{\mathcal{A}_m}{\mathcal{G}_m}$;
where  $\chi_{1N}$, $\chi_{2N}$, $\chi_{3N}$, $\Theta_{1N}$, $\Theta_{2N}$, $\Theta_{3N}$, $\varpi_{1N}$, and $\varpi_{2N}$  have been defined in Lemmas \ref{CDF_of_gammaUm_sF}, \ref{CDF_of_gammaUN_sN}, and \ref{CDF_of_gammaUN_sC}. 
\end{theorem}
\begin{proof}
	See Appendix \ref{Proof_of_ABLER_UN}.	
\end{proof}

\begin{theorem}
The approximated expression for e2e average BLER at $U_F$ can be derived as
\begin{small}
\begin{align}\label{final_of_varepsilon_UF_SF}
	\varepsilon_{U_F}\approx1+\chi_{1F}(\psi_F){\rm Ei}(-\chi_{2F}(\psi_F))\exp(\chi_{3F}(\psi_F)),
\end{align}
\end{small}when $\psi_F<\frac{\mathcal{A}_m}{\mathcal{G}_m}$; otherwise, $\varepsilon_{U_F}=1$; $\chi_{1F}$, $\chi_{2F}$, and $\chi_{3F}$ have been defined in Lemma \ref{CDF_of_gammaUm_sF}.
\end{theorem}
\begin{proof}
	Since $	\varepsilon_{U_F}=\bar{\xi}_{U_F}^{s_F}\approx F_{\gamma_{U_F}^{s_F}}\left(\psi_F\right)$, 
	(\ref{final_of_varepsilon_UF_SF}) can be simply attained by substituting (\ref{F_gammaUm_sF}) into  the aforementioned equations.
\end{proof}

\vspace*{-0.5cm}

\section{Numerical Results}
In this section, we verify the accuracy of  theoretical expressions through Monte Carlo methods. Unless otherwise stated, the simulation parameters are set as follows   \cite{zheng2022overlay,khan2022energy,zhu2024rate,li2021hardware,vu2021performance}: Power allocation coefficients $a_N=0.3$ and $a_F=0.7$; time instant index $t=2$; carrier frequency $f_c=5.9$ GHz; transmitted symbol duration $T_s=0.2$ ms; vehicle speed $v=70$ km/h; blocklength of different streams $\mathcal{L}_{s_i}=\mathcal{L}_s=500$ bits; $R_{s_N}=R_{s_F}=0.1$ bit/s/Hz and $R_{s_C}=0.005$ bit/s/Hz; variances of different channels $\Omega_{RN}=20$, $\Omega_{RF}=5$, $\Omega_{RB}=1$, $\Omega_{BN}=1.5$, and $\Omega_{BF}=0.5$; error levels of different channels $\Omega_e=\Omega_{\epsilon}=0.001$; reflection efficiency $\beta=0.45$.

In Fig. \ref{abler_vs_SNR}, we present the average BLERs versus the transmit SNR $\gamma$ with different vehicle speeds $v$. The theoretical results of the average BLER performance exhibit good
agreement with the simulated ones. As seen from Fig. \ref{abler_vs_SNR}, all the curves steadily decrease  as $\gamma$ increases.  It is evident that higher speed
levels degrade the performance of both users, since mobility can seriously interfere the estimated CSI at $t$-th time instant.  On the other hand, we find that the BLER of $u_N$ is always worse than that of $u_F$, which is determined by the different SIC processes for the two users. Specifically, $u_N$ needs to successively decode three signals, while $u_F$ only needs to decode its own signal, leading to higher BLER for $u_N$.

Fig. \ref{abler_vs_beta} illustrates the relationship between average BLER and the reflection efficiency $\beta$ with different transmit SNR $\gamma$  levels.  It's evident that larger $\gamma$ values correspond to larger SINRs, resulting in a decrease in average BLER as $\gamma$ increases.  Besides, regarding the average BLER of $u_F$, it increases with increasing $\beta$ due to larger interference caused by the backscatter link. Conversely, for the average BLER of $u_N$, it initially decreases to an optimal point and then increases with increasing $\beta$. This is because with a smaller $\beta$, the power of $s_C$ backscattered by BD is lower, causing larger $\bar{\xi}_{U_N}^{s_C}$, and finally  resulting in poor BLER performance. However, as $\beta$ becomes larger, the enhancement of $\bar{\xi}_{U_N}^{s_C}$ cannot eliminate the growth of interference caused by the backscatter link in $\bar{\xi}_{U_N}^{s_N}$ and $\bar{\xi}_{U_N}^{s_F}$. Therefore, choosing an appropriate reflection efficiency $\beta$ is crucial.

The average BLERs of different schemes versus blocklength $\mathcal{L}_s$ when $\gamma=15$ dB for different cases is plotted
in Fig. \ref{abler_vs_blocklength}, with $\Omega_e=\Omega_{\epsilon}=0.001$ for Case 1 and $\Omega_e=\Omega_{\epsilon}=0.01$ for Case 2. It can be observed that with increasing blocklength, the average BLER decreases.
Clearly, AmBC-NOMA always
exhibits better average BLER performance than the OMA counterpart, corroborating its capability of providing higher SE. Therefore, employing AmBC-NOMA in high-mobility
communication is highly beneficial.
Additionally, we find that  higher levels of CSI error leads to a deterioration in the BLER performance. Hence, maintaining accurate CSI in high-mobility scenario is a critical requirement for achieving better BLER performance.

\section{Conclusions}
In this paper, we have proposed a two-user AmBC-NOMA-assisted high-mobility V2X system  over time-selective fading channels under CSI imperfections. The derived theoretical expressions for the
BLERs have effectively characterized the system's performance. Simulation results have validated the accuracy of  theoretical results and showed the impacts of different transmit SNRs, CSI imperfections,  vehicle speeds, reflection efficiency coefficients, and blocklength  on the BLER.  Finally, simulation results have shown that the proposed AmBC-NOMA system outperforms the AmBC-OMA counterpart.

\begin{appendices}
	\setcounter{equation}{0}
	\renewcommand{\theequation}{\thesection.\arabic{equation}}
	
\vspace*{-0.5cm}	
	\section{Proof of Lemma \ref{CDF_of_gammaUm_sF}}\label{Proof_of_CDF_of_gammaUm_sF}
	\setcounter{equation}{0}
Apparently, when $u\geq\frac{\mathcal{A}_m}{\mathcal{G}_m}$,  $F_{\gamma_{U_m}^{s_F}}(u)=1$, thus we turn to the derivation when $u<\frac{\mathcal{A}_m}{\mathcal{G}_m}$.
The CDF of $\gamma_{U_m}^{s_F}$ 
can be derived as
\begin{small}
\begin{align}
	&	F_{\gamma_{U_m}^{s_F}}(u)=\Pr(\gamma_{U_m}^{s_F}\leq u)=1-\Pr(\gamma_{U_m}^{s_F}> u)\nonumber\\
	&\!\!=\!\!1\!\!-\!\!\int_{0}^{\infty}\!\!\int_{0}^{\infty}\!\!\int_{\kappa_1}^{\infty}\!\!f_{|\hat{h}_{Rm}(1)|^2}(x)f_{|\hat{h}_{Bm}(1)|^2}(y)f_{|{h}_{RB}|^2}(z)\text{d}x\text{d}y\text{d}z\nonumber\\
	&\!\!\overset{(a)}{=}\!\!1\!\!-\!\!\int_{0}^{\infty}\!\!\frac{\lambda_{Rm}}{\kappa_2\hat{\Omega}_{Bm}+\hat{\Omega}_{Rm}}\exp\left(\!-\frac{\kappa_3}{\hat{\Omega}_{Rm}}\!\right)\!\!f_{|{h}_{RB}|^2}(z)\text{d}z
\end{align}
\end{small}where $(a)$	can be arrived by using [\citen{gradshteyn2014table}, Eq. (3.310)], $	\kappa_1=\frac{u\gamma(\Omega_{\xi Rm}\!+\!\mathcal{C}_m|h_{RB}|^2|\hat{h}_{Bm}(1)|^2\!+\!|h_{RB}|^2\mathcal{D}_m)+u}{(\mathcal{A}_m-u\mathcal{G}_m)\gamma}$, $\kappa_2=\frac{u\mathcal{C}_mz}{(\mathcal{A}_m-u\mathcal{G}_m)}$, and $\kappa_3=\frac{u(\sigma^2+ z\mathcal{D}_m+1/\gamma)}{(\mathcal{A}_m-u\mathcal{G}_m)}$. Finally, by applying [\citen{gradshteyn2014table}, Eq. (3.352.4)], we
arrive at (\ref{F_gammaUm_sF}), completing the proof.
	\section{Proof of Lemma \ref{CDF_of_gammaUN_sC}}\label{Proof_of_CDF_of_gammaUN_sC}
\setcounter{equation}{0}
Based on (\ref{define_gammaUNSC}), the CDF of $\gamma_{U_N}^{s_C}$ can be derived as
\begin{small}
\begin{align}
	&	F_{\gamma_{U_N}^{s_C}}(u)=\Pr(\gamma_{U_N}^{s_C}\leq u)=1-\Pr(\gamma_{U_N}^{s_C}\geq u)\nonumber\\
	&=1-\int_{0}^{\infty}\int_{\kappa_4}^{\infty}f_{|\hat{h}_{BN}(1)|^2}(x)f_{|{h}_{RB}|^2}(y)\text{d}x\text{d}y\nonumber\\
	&=1-\frac{\exp(-\kappa_6)}{\Omega_{RB}}\int_{0}^{\infty}\exp\left(-\frac{\kappa_5}{y}-\frac{y}{\Omega_{RB}}\right)\text{d}y,
\end{align}
\end{small}where $	\kappa_4\triangleq\frac{u\gamma(\Omega_{\xi RN}+|h_{RB}|^2\mathcal{D}_N)+u}{|h_{RB}|^2\mathcal{C}_N\gamma}$, $\kappa_5\triangleq\frac{u(\gamma\Omega_{\xi RN}+1)}{\hat{\Omega}_{BN}\mathcal{C}_N\gamma}$, and $\kappa_6=\frac{u\mathcal{D}_N}{\hat{\Omega}_{BN}\mathcal{C}_N}$.
Finally, by applying [\citen{gradshteyn2014table}, Eq. (3.471.9)], we
arrive at (\ref{F_gammaUN_sC}), completing the proof.
	\section{Proof of Theorem \ref{ABLER_UN}}\label{Proof_of_ABLER_UN}
\setcounter{equation}{0}
For convenience, let 
\begin{small}
	\begin{subequations}
		\begin{align}	
			\mathcal{F}_1(u)=&\chi_{1N}(u){\rm Ei}(-\chi_{2N}(u))\exp(\chi_{3N}(u)), \label{F1u}\\
			\mathcal{F}_2(u)=&\Theta_{1N}(u){\rm Ei}(-\Theta_{2N}(u))\exp(\Theta_{3N}(u)),\label{F2u}\\
			\mathcal{F}_3(u)=&-\varpi_{1N}(u)K_1(\varpi_{1N}(u))\exp(-\varpi_{2N}(u)).\label{F3u}
		\end{align}
	\end{subequations} 
\end{small}Apparently, we have $\bar{\xi}_{U_N}^{s_F}=1+\mathcal{F}_1(u)$, $\bar{\xi}_{U_N}^{s_N}=1+\mathcal{F}_2(u)$, and $\bar{\xi}_{U_N}^{s_C}=1+\mathcal{F}_3(u)$.	Utilizing the aforementioned equations and 
after performing some mathematical operations, $\varepsilon_{U_N}$ can be converted to
\begin{small}
\begin{align}
	\varepsilon_{U_N}&=\bar{\xi}_{U_N}^{s_F}+(1-\bar{\xi}_{U_N}^{s_F})\bar{\xi}_{U_N}^{s_N}+(1-\bar{\xi}_{U_N}^{s_F})(1-\bar{\xi}_{U_N}^{s_N})\bar{\xi}_{U_N}^{s_C},\nonumber\\
	&=1+\mathcal{F}_1(\psi_F)\mathcal{F}_2(\psi_N)\mathcal{F}_3(\psi_C)
\end{align}
\end{small}Then, by substituting (\ref{F1u}),  (\ref{F2u}), and (\ref{F3u}) into the above equation, we can derive the approximated expressions of $\varepsilon_{U_N}$, completing the proof.
\end{appendices}

\vspace*{-0.5cm}
\bibliographystyle{IEEEtran}
\bibliography{BAC_NOMA_high_mobility_BLER.bib}
\end{document}